\documentclass[a4paper,UKenglish,cleveref, autoref]{lipics-v2019}

\hideLIPIcs
\title{Efficient Estimation of Graph Trussness} %TODO Please add
\author{Alessio Conte}{University of Pisa}{conte@di.unipi.it}{}{}
\author{Roberto Grossi}{University of Pisa}{grossi@di.unipi.it}{}{}
\author{Andrea Marino}{University of Florence}{andrea.marino@unifi.it}{}{}
\author{Luca Versari}{University of Pisa and Google Research Zurich}{luca.versari@di.unipi.it}{}{}

\newcommand{\shortauthors}{A. Conte,   R. Grossi, A. Marino, L. Versari}

\authorrunning{\shortauthors}%TODO mandatory. First: Use abbreviated first/middle names. Second (only in severe cases): Use first author plus 'et al.'

\Copyright{A. Conte,   R. Grossi, A. Marino, L. Versari}%TODO mandatory, please use full first names. LIPIcs license is "CC-BY";  http://creativecommons.org/licenses/by/3.0/

%TODO mandatory; please add comma-separated list of keywords

\category{}%optional, e.g. invited paper

\relatedversion{}%optional, e.g. full version hosted on arXiv, HAL, or other respository/website
\supplement{}%optional, e.g. related research data, source code, ... hosted on a repository like zenodo, figshare, GitHub, ...

\nolinenumbers %uncomment to disable line numbering

\EventEditors{John Q. Open and Joan R. Access}
\EventNoEds{2}
\EventLongTitle{42nd Conference on Very Important Topics (CVIT 2016)}
\EventShortTitle{CVIT 2016}
\EventAcronym{CVIT}
\EventYear{2016}
\EventDate{December 24--27, 2016}
\EventLocation{Little Whinging, United Kingdom}
\EventLogo{}
\SeriesVolume{42}
\ArticleNo{23}

\usepackage[utf8]{inputenc}

\usepackage{algorithm}
\usepackage{float}

\usepackage{url}
\usepackage{xspace}
\usepackage{amsmath}
\usepackage{amssymb}
\usepackage{etoolbox}
\usepackage{tikz}

\usepackage{apxproof}

\usepackage{graphicx}
\usepackage{color}

\usepackage{xcolor}
\usepackage[normalem]{ulem} % use normalem to protect \emph

\usepackage{url} % per il repo col codice (vedi inizio sezione experiments)
\usepackage{multirow}
\usepackage{balance}
\usepackage{graphicx}

\usepackage{booktabs} % For formal tables

\begin{document}

%\newcommand\hl{\bgroup\markoverwith    {\textcolor{yellow}{\rule[-.5ex]{2pt}{2.5ex}}}\ULon}

%\newenvironment{proof}[1][Proof]{\begin{trivlist}
%\item[\hskip \labelsep {\bfseries #1}]}{\hfill \fbox{}\end{trivlist}}

% \newtheorem{theorem}{Theorem}
% \newtheorem{lemma}{Lemma}
% \newtheorem{definition}{Definition}

%%%%%%%%%%%% OUR MACROS %%%%%%%%%%%%%%%%%
\def\poly{\operatorname{poly}}
\def\polylog{\operatorname{polylog}}
\newcommand{\ktruss}{$k$-truss\xspace}
\newcommand{\ktrusses}{$k$-trusses\xspace}
\newcommand{\truss}{\ensuremath{\mathit{truss}}\xspace}
\newcommand{\degr}{\delta\xspace}

\newcommand{\tri}{^{^{\triangle}}\xspace}
\newcommand{\gtri}{G\tri\xspace}
\newcommand{\gtrip}{G^{^{\triangle}}_p\xspace}
\newcommand{\gtripprime}{G'^{^{\triangle}}_p\xspace}
\newcommand{\minsupp}{\textsc{min-sup}\xspace}
\newcommand{\supp}{\textsc{sup}\xspace}
\newcommand{\support}{\ensuremath{\supp}\xspace}
\newcommand{\whp}{w.h.p.\xspace}
\newcommand{\mirr}[1]{^{\times#1}}
\newcommand{\gmirr}[1]{{G\mirr{#1}}}
\newcommand{\epsp}{\epsilon'}

\maketitle

\begin{abstract}
A $k$-truss is an edge-induced subgraph $H$ such that each of its edges belongs to at least $k-2$ triangles of $H$.
This notion has been introduced around ten years ago in social network analysis and security, as a form of cohesive subgraph that is rich of triangles and less stringent than the clique. The \emph{trussness} of a graph is the maximum $k$ such that a $k$-truss exists. 

The problem of computing $k$-trusses has been largely investigated from the practical and engineering point of view. On the other hand, the theoretical side of the problem has received much less attention, despite presenting interesting challenges.
The existing methods share a common design, based on iteratively removing the edge with smallest support, 
where the support of an edge is the number of triangles containing it. 

The aim of this paper is studying algorithmic aspects of graph trussness.
While it is possible to show that the time complexity of computing exactly the graph trussness and that of counting/listing all triangles is inherently the same, we provide efficient algorithms for estimating its value, under suitable conditions, with significantly lower complexity than the exact approach. 
In particular, we provide a $(1 \pm \epsilon)$-approximation algorithm that is asymptotically faster than the exact approach, on graphs which contain $\omega(m \polylog(n))$ triangles, and has the same running time on graphs that do not. 
For the latter case, we also show that it is impossible to obtain an approximation algorithm with faster running time than the one of the exact approach when the number of triangles is $O(m)$, unless well known conjectures on triangle-freeness and Boolean matrix multiplication are false.
\end{abstract}

\begin{figure*}
    \centering
%    \hspace*{50pt}
    \raisebox{-0.5\height}{\includegraphics[height=0.3\textwidth]{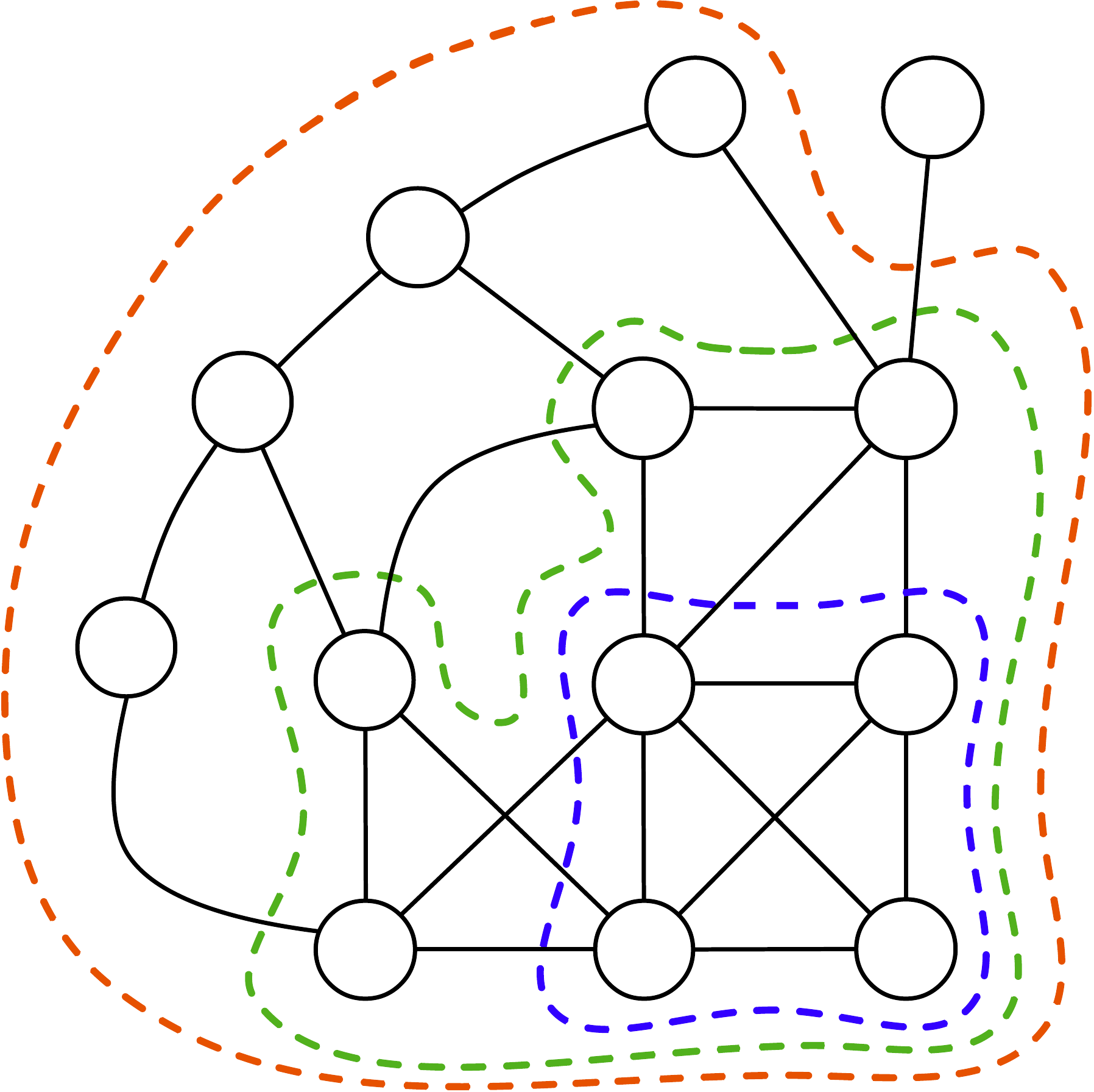}}
    \hspace{1in}
    \raisebox{-0.5\height}{\begin{tikzpicture}
    \node[] at (0, 0) {\includegraphics[height=0.3\textwidth]{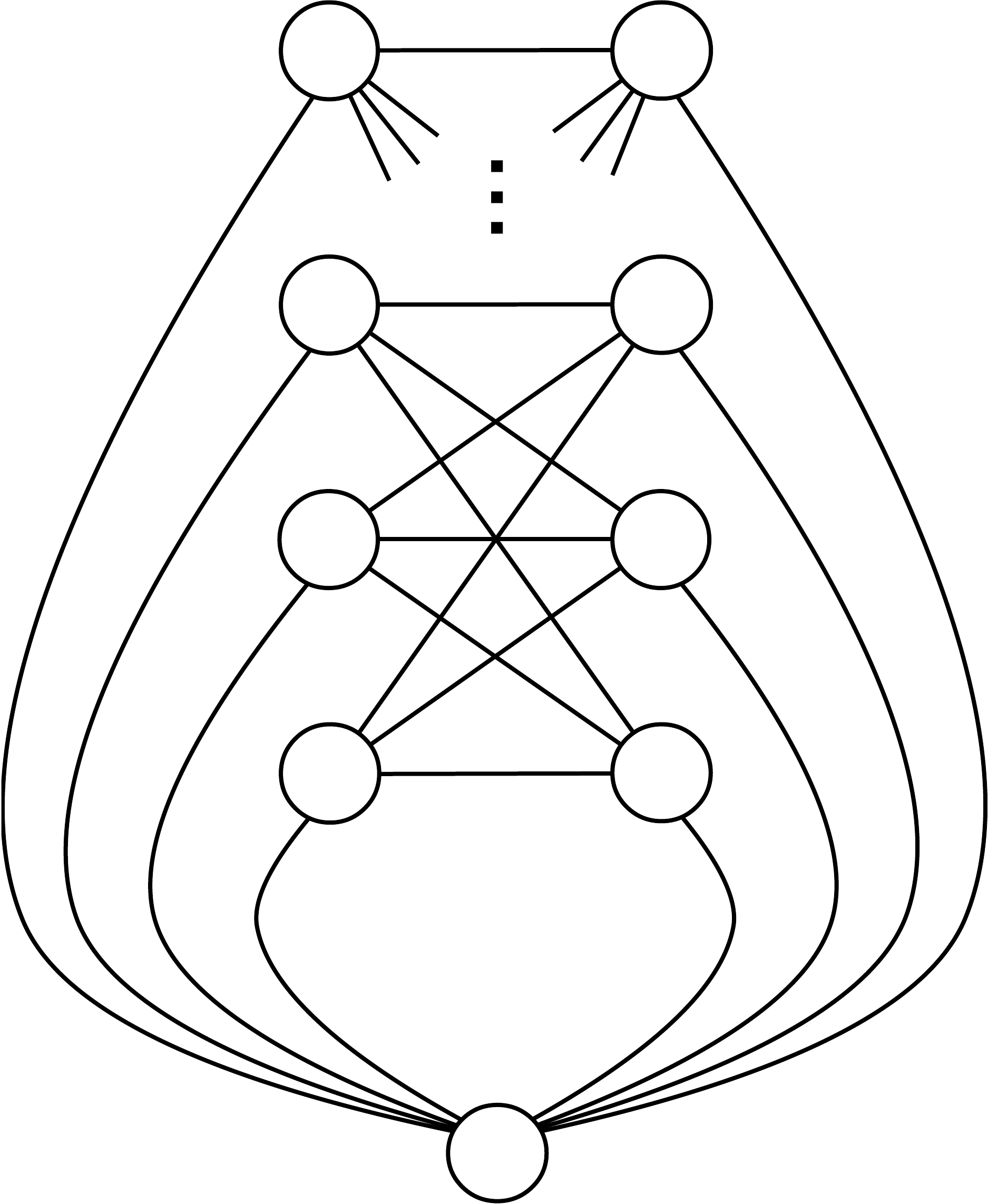}};
    \node[] at (0, -2.4) {\scalebox{.9}{$v_x$}};
    \end{tikzpicture}}
    \caption{Left: a graph containing (dashed lines, starting from the outside) a $2$-core (orange), $1$-truss (green) and $4$-clique (blue); a $4$-clique is also a $2$-truss --according to the $0$-based definition--, so the graph trussness is $t_G=2$. Note that a \ktruss is a subgraph of a $(k+1)$-core, and that a $k$-clique is a $(k-2)$-truss. Right: a graph (complete bipartite plus an extra node $v_x$) with trussness $t_G=1$, arboricity $\alpha_G = \Theta(n)$, and $T_G = \Theta(n^2)$ triangles.
    }
    \label{fig:example}
    \label{fig:example-trussness}
    \label{fig:truss-decomposition}
\end{figure*}

\section{Introduction}
\label{sec:introduction}

Consider an undirected graph $G=(V(G),E(G))$ with $n = |V(G)|$ nodes and $m=|E(G)|$ edges, where $N_G(v)$ represents the neighborhood of a node $v$.\footnote{We assume wlog that $G$ does not contain isolated nodes, thus its size is $O(n+m) = O(m)$.} Triangles in $G$ are popular patterns studied in social networks to identify cohesive subgraphs. Recall that a triangle is a set of three pairwise-connected nodes $u,v,z$ (hence, egdes $\{u,v\}, \{v,z\}, \{z,u\}$ belong to that triangle). For each edge $e =  \{u,v\}$ in $G$, its \emph{support} $\supp_G(e) = |N_G(u) \cap N_G(v)|$ is the number of triangles to which $e$ belongs.

Among the cohesive subgraphs based on triangles, \ktrusses have quickly spread in security and social sciences, as they enforce the presence of many triangles~\cite{cohen2008trusses}. For an integer $k \geq 0$, we define the \emph{\ktruss} of $G$ as the maximal (edge-induced) subgraph $H$ of $G$ such that each edge $e$ of $H$ belongs to at least $k$ triangles of $H$.~\footnote{The original definition in~\cite{cohen2008trusses} is $2$-based, i.e., assumes that $k \geq 2$ and requires that $\supp_H(e) \geq k-2$, so a $k$-clique is a \ktruss. However, it is more convenient algorithmically, as we will see, to adopt our equivalent $0$-based definition.} Specifically, $H = (V(H), E(H))$ where $E(H) \subseteq E(G)$, and $V(H) = \{ x \in V(G) \mid \{x,y\} \in E(H)\}$, such that $\supp_H(e) \geq k$ for every $e \in E(H)$. Visual examples are given in Figure~\ref{fig:truss-decomposition} (left), comparing \ktrusses with the known notions of $k$-cores (i.e., each node in the subgraph has degree at least $k$) and $k$-cliques (i.e., the $k$ nodes in the subgraph are pairwise connected).

Over the years, the notion of \ktruss has become popular in community detection, and is gaining momentum for purposes other than security~\cite{cohen2008trusses,wang2012truss,huang2014querying,chen2014distributed,smith8091049shared,Voegele8091037,kabir2017parallel,Wu18ConsumerGrade,davis2018graph}, providing a remarkable benchmark (along with triangle counting and listing) to test new ideas, such as in the MIT/Amazon/IEEE GraphChallenge~\cite{kabir2017parallel,gchallenge,smith8091049shared,pearce2018k,date2017collaborative}. It has also been considered under different names, such as $k$-dense  subgraph~\cite{saito2008extracting}, triangle $k$-core~\cite{zhang2012extracting}, $k$-community~\cite{verma2013network}, and $k$-brace~\cite{ugander2012structural}.

Many existing highly-engineered solutions~\cite{smith8091049shared,kabir2017parallel,gchallenge} are sophisticated implementations of the same basic algorithmic idea presented in~\cite{cohen2008trusses}, and sometimes referred to as \emph{peeling}: the \ktruss of $G$ is obtained by recursively deleting edges with support smaller than $k$ in the residual graph.\footnote{In a similar fashion to $k$-cores, obtained by recursively removing nodes with less than $k$ neighbors.} As a matter of fact, these algorithms compute a \emph{truss decomposition}, namely, for each edge $e$ in $G$, they find its trussness $t_G(e)$ which is the largest $k$ such that a $k$-truss contains $e$. 

\subparagraph*{Problem studied}
In this paper we study the problem of approximating the \emph{trussness} $t_G$ of $G$, defined as the maximum $k$ such that there exists a \ktruss in $G$. Note that $G$ is \emph{triangle-free} if and only if $t_G=0$. As $t_G = \max_{e \in E(G)} t_G(e)$, we can clearly compute exactly the trussness of $G$ from its truss decomposition. However, we observe in Appendix~\ref{app:lowerb} that 
(1)~not only the truss decomposition, and thus $t_G$, can be computed in $O(m \, \alpha_G)$ time and $O(m)$ space, where the \emph{arboricity} $\alpha_G$ is the minimum number of forests into which the edges of $G$ can be partitioned,\footnote{A bound different but equivalent to $O(m \, \alpha_G)$ for the truss decomposition is also shown in~\cite{burkhardt2018bounds}.} 
but that 
(2)~a conditional lower bound exists, suggesting that $O(m \, \alpha_G)$ time is difficult to improve for an exact computation of the value of $t_G$. 

Because of~(1) and~(2), it is an interesting algorithmic question to see if the trussness $t_G$ can be approximated in less than $O(m \, \alpha_G)$ time without resorting to the exact truss decomposition.

We say that $\tilde{t}_G$ is an $r$-approximation for $t_G$, where $r>1$ is a constant, if the relation 
$\frac{t_G}{r} \leq \tilde{t}_G \leq r \, t_G$
holds. For the special and interesting case $r = 1 + \epsilon$ with $0 < \epsilon < 1$, this implies the condition $(1 - \epsilon) \, t_G \leq \tilde{t}_G \leq (1 + \epsilon) \, t_G$, and we call $\tilde{t}_G$  a $(1 \pm \epsilon$)-\emph{approximation}.\footnote{In turn this implies that $(1 - \epsilon) \, (t_G+2) \leq \tilde{t}_G + 2 \leq (1 + \epsilon) \, (t_G+2)$. As the trussness originally defined in~\cite{cohen2008trusses} is $t_G^* = t_G+2$, we observe that $\tilde{t}_G + 2$ is a $(1 \pm \epsilon$)-approximation for $t_G^*$
as well.} Given graph $G$, we want to compute a $(1 \pm \epsilon$)-approximation $\tilde{t}_G$ faster than computing $t_G$.

\subparagraph*{Results} 
Let $T_G$ be the number of triangles in $G$. We show that we can improve the $O(m \, \alpha_G)$-time bound when $T_G = \omega(m \polylog(n))$ (for all $\polylog(n)$). 
This follows from our general result that, for any $0<\epsilon<1$, a $(1 \pm \epsilon)$-approximation of the trussness $t_G$ can be obtained, with high probability, in expected time $O\left(\epsilon^{-3}\min\left\{\frac{m \log m}{T_G+1}, 1\right\} m \, \alpha_G \, \log (t_G+2)\right)$\footnote{We use $T_G+1$ and $t_G+2$ in place of $T_G$ and $t_G$ simply to preserve coherence of the bound when $T_G=0$ or $t_G\le 1$.} and space $O(\epsilon^{-2}m \log m)$ (Theorem~\ref{thm:estimate}).

Looking at the above bound, we also observe that in case $T_G$ is smaller, e.g. $T_G= \Theta(m)$, our approximation algorithm substantially has the same time cost $O(m \, \alpha_G)$ as the exact computation of $t_G$. 
However, we prove a matching conditional lower bound for this case, stating that there is no combinatorial algorithm which approximates $t_G$ within a multiplicative factor or an additive term taking $\tilde{o}(m \, \alpha_G)$ time, unless unless there exists a truly subcubic ``combinatorial'' algorithm for Boolean Matrix Multiplication (BMM)~\cite{williams2010subcubic}.\footnote{A definition of ``combinatorial'' algorithm, and a more detailed discussion of this conditional lower bound, which also includes the problems of (i)~detecting if a graph is triangle-free and (ii)~listing up to $n^{3-\delta}$ triangles in a graph for constant $\delta > 0$, are given in Section~\ref{sub:lower-bound-T-approx}. We remark for the moment that, as a rule-of-the-thumb, an algorithm is typically combinatorial if it does not use fast matrix multiplication.}
Indeed we prove that such an algorithm would improve the complexity of BMM and triangle-freeness, breaking the well known cubic lower bound (see~\cite{williams2010subcubic}) as well as the 3-SUM conjecture~\cite{Kopelowitz:2016:3sum}.

Under the same conditions, i.e., the non-existence of a truly subcubic combinatorial algorithm for BMM, we also show that no combinatorial algorithm that provides an approximation by either a multiplicative factor or an additive term, may run in time $O(m (t_G+1))$.

If, on the other hand, we move away from the scope of combinatorial algorithms and use matrix multiplication, where $\omega$ denotes, as in~\cite{williams2010subcubic}, the smallest real number such that an $(n \times n)$-matrix multiplication %over an arbitrary ring 
can be computed in $O(n^\omega)$, we show that a $(2+\epsilon)$-approximation of the trussness $t_G$ can be computed in either $O(\epsilon^{-1}n^{\omega}\log{\frac mn})$ or $O(\epsilon^{-1}m^{1+\frac{\omega-1}{\omega+1}})$ time (Theorem~\ref{thm:approx-matrix-multiplication}).

The exact computations of $t_G$ and $T_G$ are thus related since they share the same upper and conditional lower bounds. However, for the approximate computation, the trussness exhibits a hybrid behavior: when 
$T_G$ is polylogarithmically close to the number of edges in $G$, or smaller, approximation is as hard as listing the triangles in $G$; when $T_G$ is sufficiently large, instead, it can be much faster, as is the case for triangle counting. In particular, for $T_G = \Theta(m)$ and constant $\epsilon$, finding a $(1  \pm \epsilon)$-approximation for the trussness with a combinatorial algorithm is provably harder (unless BMM is truly subcubic) than finding a $(1  \pm \epsilon)$-approximation for triangle counting: the former \emph{cannot} be done in significantly less  $O(m \, \alpha_G)$ time (Theorem~\ref{thm:conditional-bound-T-approx}), whereas the latter takes $O\left(\left(\frac{n}{T_G^{1/3}+1}+\min\{m,\frac{m^{3/2}}{T_G}+1\}\right)\cdot \poly(\log n,\frac{1}{\epsilon})\right) = O\left(n^{2/3}+ \sqrt{m} \polylog(n)\right) = \tilde{o}(m)$ expected time
using the result in~\cite{eden2017approximately}.
 
We also provide further satellite results on triangle enumeration, trussness and truss decomposition in the rest of the paper.

\subparagraph*{Related work}
The seminal paper by Cohen~\cite{cohen2008trusses} presents the first algorithm to compute a \ktruss and a variety of highly optimized implementations have tuned that algorithm: sequential for massive networks~\cite{wang2012truss} in $O(m^{3/2})$ time, improved in~\cite{rossi2014fast}, parallel~\cite{smith8091049shared,kabir2017parallel}, distributed~\cite{chen2014distributed}, using data-centric models~\cite{Voegele8091037}, or parallel matrix multiplication on GPUs~\cite{bisson8091034static,green2017quickly}. Some of these were awarded in the 2017 and 2018 editions of the GraphChallenge~\cite{gchallenge}.

Furthermore,~\cite{huang2014querying} considered querying $k$-trusses under edge deletion, and~\cite{akbas2017truss} proposed an index computable in $O(m^{3/2})$ time, and maintainable on dynamic graphs, that permits to retrieve the $k$-truss containing a vertex $v$ in output-sensitive time.

A notion of trussness has also been given for probabilistic graphs~\cite{huang2016truss}, and uncertain graphs~\cite{zou2017truss}.
The \ktruss notion has been also extended with additional constraint of co-location, i.e., requiring connectivity and bounded diameter within the \ktruss~\cite{chen2018maximum}. 

From the algorithmic perspective, the problem has received much less attention. Recently,~\cite{burkhardt2018bounds} studied novel combinatorial properties on \ktrusses and proposed an exact algorithm for computing the truss-decomposition of a graph running with time $O(m\bar{\delta}(G))$, where $\bar{\delta}(G)$ is the average degeneracy of $G$. Moreover, the same paper provided an exact randomized algorithm which takes a parameter $K_{max}$ as input, and with high probability computes the $k$-truss-decomposition for all $k\leq K_{max}$, using matrix multiplication.

In this paper, we study approximation algorithms and conditional lower bounds for computing trussness, using several ideas: we introduce gadgets of controlled size to suitably amplify the trussness while preserving other parameters; we study several properties of a related hypergraph which equivalently represents the triangles; we provide a simple technique for triangle sampling based on wedge sampling~\cite{seshadhri2013triadic,schank2004approximating,etemadi2016efficient}, to name a few.

\section{Approximating the Trussness}

We describe how to obtain a $(1 \pm \epsilon)$-approximation of the trussness $t_G$ \whp using a combinatorial algorithm. Let $T_G$ be the number of triangles in $G$, observing that $T_G = O(m \, \alpha_G)$. 

\begin{theorem}
The trussness of a graph $G$ can be $(1\pm\epsilon)$-approximated  \whp, for any $0<\epsilon<1$, in expected time $O\left(\epsilon^{-3}\min\left\{\frac{m \log m}{T_G+1}, 1\right\} m \, \alpha_G \, \log (t_G+2)\right)$ and space $O(\epsilon^{-2}m \log m)$.
\label{thm:estimate}
\end{theorem}

In the rest of section, we give an algorithm which meets the requirements of Theorem~\ref{thm:estimate}, whose structure is divided into main blocks according to the following roadmap.

\begin{enumerate}
\item In Section~\ref{sec:tools} we firstly recall the concepts of degeneracy and truss order; in order to approximate the trussness, and introduce a notion of \emph{approximate} truss order.
\item In Section~\ref{sec:estimate-trussness}, we show how to link the truss order of $G$ to its trussness: we add gadgets to $G$, whose trussness is known, and deduce the trussness of the edges of $G$ by looking at their position in the truss order with respect to the edges of the gadget.
When the truss order is not exact, understanding the trussness of $G$ is more complex. Nonetheless, we show that a suitable approximate order can also be used to approximate $t_G$.

\item In Section~\ref{sec:computix-apx-truss-order}, we focus on efficiently computing an approximate truss order of $G$. In particular, in Section~\ref{sec:sampling} we introduce a transformation which maps the graph $G$ into a hypergraph $\gtri$, whose degeneracy turns out to be exactly the trussness of the original graph. We build a sample of $\gtri$, called $\gtrip$ obtained by sampling triangles uniformly at random from $G$ with probability $\Omega(p)$ for a suitable $p$. 
Finally, in Section~\ref{sec:degtrip} we show that \whp we can get an approximate degeneracy ordering of $\gtri$ using $\gtrip$ which also induces an approximate truss order. % in $G$.
\end{enumerate}

\subsection{Degeneracy order and truss order}
\label{sec:tools}

Before getting into the details of the approximation algorithm, we need to introduce some concepts that will help us to achieve our goal.

The degeneracy of a graph $G$ is a well-known sparsity measure~\cite{DBLP:conf/icalp/ConteGMV16,EppsteinLS13,wasa2014efficient,Matula:1983}, and is indicated here as $d_G$. It is defined as the largest integer $k$ such that there exists a $k$-core in $G$. It known that $d_G \le 2\cdot \alpha_G$,\footnote{This is proven by the fact that $G$ can be partitioned in $\alpha_G$ forests, so it has $\le \alpha_G (n-1)$ edges and an average degree $\le 2 \alpha_G$. As any subgraph of $G$ has at most the same arboricity, any subgraph has a vertex of degree at most $2\alpha_G$, so $d_G\le 2\alpha_G$.} and thus $d_G = O(\sqrt{m})$. We will use the \emph{degeneracy order} of the nodes in a graph, which can be defined as the order obtained by repeatedly removing the node of minimum (residual) degree from $G$. It takes $O(m)$ time to obtain a degeneracy ordering, and the $d_G$ is equal to the maximum among the residual degrees (called \emph{forward} degrees in the following).

We introduce similar notions for \ktrusses. 
A \emph{truss order} of $G$ is an ordering $\langle e_1,\ldots,e_m\rangle$ of its edges obtained by repeatedly removing the edge of smallest (residual) support in $G$, breaking ties arbitrarily. For an edge $e_i$ in the order, we call \emph{forward triangles} the triangles that $e_i$ forms using two other edges chosen from the residual graph (i.e., induced by $e_i,\ldots, e_m$), and \emph{forward support} their number $t_G(e_i)$. We observe that $t_G = \max_{i=1}^m t_G(e_i)$ by definition. In the truss order, the edges of $E(G)$ are numbered consistently with their trussness, i.e, $t_G(e_i) < t_G(e_j)$ implies $i < j$.

In the following, for a given order of the edges $\langle e_1,\ldots,e_m\rangle$, let $G_{\geq e_i}$ be the graph induced by the edges $e_i, \ldots, e_m$,\footnote{Whenever needed, we define $G_{ \geq v_j}$ as the graph induced by nodes $v_j, \ldots, v_n$ in the order $\langle v_1,\ldots,v_n\rangle$.} and $\minsupp(G_{\ge e_i}) = \min_{j \geq i} \supp_{G_{\ge e_i}}(e_j)$ be the minimum support of an edge in $G_{\ge e_i}$.
Note that a truss order guarantees that $t_G(e_i) = \supp_{G_{\ge e_i}}(e_i) = \minsupp(G_{\ge e_i}) \leq t_G$.

We introduce a definition of approximate truss order, inspired by the approximate degeneracy order introduced in~\cite{Farach-ColtonT14}, that can be computed faster than $t_G$ under certain conditions. 
For $\epsilon > 0$, we say that $\langle e_1,\ldots,e_m\rangle$ is a $(1+\epsilon)$-\emph{approximate truss order} of $G$ if, for every edge $e_i$, the number of forward triangles to which it belongs is upper bounded as
\begin{equation}
\label{eq:approx-order}
 \supp_{G_{\geq e_i}}(e_i) \leq \max\left\{ \frac{T_G}{m}, \: (1+\epsilon) \,\minsupp(G_{\ge e_i})\right\}   
\end{equation}
where $T_G$ is the number of triangles in $G$.
At this point it is relevant to mention how $T_G$ relates to $t_G$ by an extension of Nash-Williams' result~\cite{Nash-Williams} to trussness, proven in~\cite{conte2020truly}:

\begin{theorem}[(From~\cite{conte2020truly})]\label{thm:nashwill}
Given an undirected graph $G$ with trussness $t_G$, 
let $T_S$ be the number of triangles and $m_S$ be the number of edges in any subgraph $S$ of $G$. \break Then 
$\max_{S\subseteq G}\frac{T_S}{m_S} \leq t_G \leq 3\max_{S\subseteq G}\frac{T_S}{m_S}$.%\prpshort
\end{theorem}

Indeed, this result implies $t_G \geq T_G/m$, meaning that in formula~\eqref{eq:approx-order} each edge in the $(1+\epsilon)$-approximate truss order has at most $t_G\cdot (1+\epsilon)$ forward triangles. 

\subparagraph*{Gadget \boldmath$\gmirr{q}$}
Given a graph $G$ we use the gadget $\gmirr{q} = (V(\gmirr{q}),E(\gmirr{q}))$, known as \emph{balanced blow-up}~\cite{HATAMI2014196} and defined as follows: Let $G$ be called $G^1$ and, for each $i=2,\ldots, q$, let $G^i = (V(G^i), E(G^i))$ be an exact copy of $G$. For each $v\in V(G)$, we call $v^i$ the corresponding node in $V(G^i)$. Thus, $V(\gmirr{q}) = V(G^1) \cup \cdots \cup V(G^q)$ contains $nq$ nodes.
For each edge $\{u,v\}\in E(G)$, connect all pairs of nodes $u^i,v^j$ with $i,j\in\{1,\ldots, q\}$. Thus $E(\gmirr{q}) = E(G^1) \cup \cdots \cup E(G^q) \cup \{\{u^i,v^j\} \mid \{u,v\} \in E(G) \mbox{ and } i \neq j\}$ contains $O(mq^2)$ edges.

The main purpose of $\gmirr{G}$ is amplifying the trussness of $G$ by a factor $q$, introducing a limited amount of new nodes, edges and triangles. More formally, we can prove that:

\begin{lemma}
\label{lem:amplify-trussnes}
Given an undirected graph $G$ with $n$ nodes, $m$ edges, trussness $t_G$, and an integer $q>1$, we can obtain a graph $\gmirr{q}$ with $|V(\gmirr{q})| = \Theta(nq)$ vertices, $|E(\gmirr{q})| = \Theta(mq^2)$ edges, and trussness $t_{\gmirr{q}} = q \, t_G$ (and $q^3\, T_G$ triangles). Given access to $G$, it is possible to implicitly navigate in $\gmirr{q}$ without materializing it. 
\end{lemma}
\begin{proof}
Let us recall the definition of $\gmirr{q} = (V(\gmirr{q}),E(\gmirr{q}))$, given $G$: Let $G$ be called $G^1$ and, for each $i=2,\ldots, q$, let $G^i = (V(G^i), E(G^i))$ be an exact copy of $G$. For each $v\in V(G)$, we call $v^i$ the corresponding node in $V(G^i)$. Thus, $V(\gmirr{q}) = V(G^1) \cup \cdots \cup V(G^q)$ contains $nq$ nodes.
For each edge $\{u,v\}\in E(G)$, connect all pairs of nodes $u^i,v^j$ with $i,j\in\{1,\ldots, q\}$. Thus $E(\gmirr{q}) = E(G^1) \cup \cdots \cup E(G^q) \cup \{\{u^i,v^j\} \mid \{u,v\} \in E(G) \mbox{ and } i \neq j\}$ contains $O(mq^2)$ edges. For any $i,j$, we call $\{u^i,v^j\}$ the mirror edges of $\{u,v\}$.

We now show that $\gmirr{q}$ has trussness $t_\gmirr{q} = q \, t_G$. Consider any triplet of distinct nodes $a^h,b^i,c^j\in \gmirr{q}$, where $h,i,j \in \{1,\ldots,q\}$. They form a triangle in $\gmirr{q}$ if and only if $a,b,c$ form a triangle in $G$, so we have $q$ choices per edge, and thus the number of triangles is increased by a factor of $q^3$. 

Furthermore, consider any edge $\{a,b\} \in E(G)$, recalling that it belongs to $\supp_G(\{a,b\})$ triangles in $G$, and let $abc$ one such triangle. Looking at edge $\{a^h,b^i\} \in E(\gmirr{q})$, we observe that it belongs to the triangles $a^hb^ic^1, \ldots, a^hb^ic^q$. In other words, $\supp_\gmirr{q}(\{a^i,b^j\}) = q \, \supp_G(\{a,b\})$. Note this property is preserved for any subgraph $H$ of $G$, that is, $\supp_{H\mirr{q}}(\{a^i,b^j\}) = q \, \supp_H(\{a,b\})$, where $H\mirr{q}$ is the subgraph of $\gmirr{q}$ made up of the \emph{mirror} edges of those in $H$. 

Since $G$ has trussness $t_G$, let $H$ be a \ktruss in $G$ with $k=t_G$. Since each edge $\{a,b\} \in E(H) \subseteq E(G)$ has support $\supp_H(\{a,b\}) \geq t_G$, each mirror edge $\{a^h,b^i\}$ has support $\supp_{H\mirr{q}}(\{a^i,b^j\}) \geq q \, t_G$, and thus $H\mirr{q}$ is a \ktruss in $\gmirr{q}$ with $k \geq q \, t_G$. Note that the edge(s) of smallest support in $H$ (i.e., with support $t_G$) will have support exactly $q \, t_G$ in $H\mirr{q}$ and thus its trussness is exactly $q \, t_G$.

On the other hand, suppose by contradiction that $\gmirr{q}$ has a \ktruss $J$, where $k > q \, t_G$: thus the edge of smallest support in $J$ has support $> q \, t_G$ in $J$. Define $J^* = H\mirr{q}$, where $E(H) = \{ \{a,b\} \mid \{a^h,b^i\} \in E(J)\} \subseteq E(G)$. For each triangle $a^h b^i x$ in $J$, there is now at least one triangle $a^{h'} b^{i'} x$ in $J^*$ by definition of $J^*$. This means that each edge $\{a^{h'}, b^{i'}\}$ in $J^*$ has support at least equal to the support of some edge $\{a^h, b^i\}$ in $J$ by construction, thus $J^*$ is also a \ktruss $J$, where $k > q \, t_G$. This is a contradiction as it implies that $H$ is a \ktruss where $k > t_G$, and thus $G$ has trussness $> t_G$.

It is worth noting that we can navigate in $\gmirr{q}$ without materializing it. We discuss five navigation operations needed in this paper, but it can be extended to more.
\begin{itemize}
\item Iteration through all the nodes or edges: $O(nq)$ and $O(mq^2)$ time.
\item Adjacency check: same cost as in $G$, as $a^h$ and $b^i$ are neighbors iff $a$ and $b$ are neighbors in $G$.
\item Adjacency list: $O(q\delta(a^h))$ time for any node $a^h$, as its neighbors can be iterated as there are $q$ copies of the neighborhood of $a$ in $G$.
\item The $i$-th neighbor: same cost as in $G$, as there are $q$ copies of the neighborhood.
\item Node degree: same cost as in $G$, as there are $q$ copies of the neighborhood.
\end{itemize}
Our algorithms using the above operations can run on $\gmirr{q}$ without computing it, if given access to $G$.
\end{proof}

\newcommand{\gin}{G_{\mathit{in}}\xspace}

\subparagraph*{Gadget \boldmath$G(x)$}
We introduce another gadget, whose size and trussness are suitably bounded.

\begin{lemma}\label{lem:spurious-cliques}
Given a graph $G$ with $n$ nodes and $m$ edges, and an integer $x \geq 0$ with
$x = O(\sqrt{m})$, let $G(x)$ be the graph obtained from $G$ by adding
$\left\lceil m / \binom{x+2}{2}\right\rceil$ disjoint $(x+2)$-cliques, called \emph{spurious cliques} (whose edges are also called spurious).
Then $|V(G(x))| = O(m)$, $|E(G(x))| = O(m)$, and $\frac{T_{G(x)}}{|E(G(x))|}
\le \min\{\frac{t_G}{2}+\frac{x}{3}, \max\{t_G, \frac{x}{3}\}\}$.
\end{lemma}
\begin{proof}
Let $\ell = \left\lceil m / \binom{x+2}{2}\right\rceil$. We observe that $G(x)$ has $T' = T_G + \ell \binom{x+2}{3}$ triangles and $m' = m + \ell \binom{x+2}{2}$ edges. Moreover, $m' \ge 2m$ and $m' \ge (2\ell-1)\binom{x+2}{2}$. Also note that $m' = \Theta(m)$, since $m \ge (\ell-1)\binom{x+2}{3}$ and $x = O(\sqrt{m})$.
As for $|V(G(x))|$, this is $n+ (x+2)\ell = O(m)$.
We have $\frac{T'}{m'} = \frac{T_G+ \ell\binom{x+2}{3}}{m+\ell\binom{x+2}{2}} \le \frac{T_G}{2m} + \frac{\ell\binom{x+2}{3}}{(2 \ell-1)\binom{x+2}{2}} = \frac{T_G}{2m} + \frac{\ell x}{3(2\ell-1)} \le \frac{t_G}{2}+\frac{x}{3}$ (as $\ell \ge 1$).
Moreover, as $\frac{a+c}{b+d} \le \max\left\{\frac ab,\frac cd\right\}$ when $a, b, c, d$ are all positive, and since $\frac{T_G}{m}\le t_G$, we also have that $\frac{T'}{m'} = \frac{T_G+ \ell\binom{x+2}{3}}{m+\ell\binom{x+2}{2}} \le \max\left\{\frac{T_G}{m}, \frac{\ell\binom{x+2}{3}}{\ell\binom{x+2}{2}}\right\} \le \max\{t_G, \frac{x}{3}\}$.
From this we have $\frac{T'}{m'} \le \min\{\frac{t_G}{2}+\frac{x}{3}, \max\{t_G, \frac{x}{3}\}\}$.
\end{proof}

\begin{algorithm}[t]
\caption{Approximating the graph trussness\label{alg:estimate}}
\small
\noindent\textbf{Input:} graph $\gin$ and $0<\epsilon<1$\\
\noindent\textbf{Output:} \whp $(1 \pm \epsilon)$-approximation $\tilde{t}_{\gin}$ of the trussness $t_{\gin}$ 
\hrule
\begin{enumerate}
    \itemsep0em 
    \item $G \gets $ disjoint union of $\gin \mirr{6}$ and a triangle\label{step:trasformation}  \quad 
    \item $\epsp \gets \frac{\epsilon}{6}$, $x\gets 1$, $\tilde{t}_G \gets 1$ \label{step:init}
    \item $\langle e_1, \ldots, e_{|E(G(x))|} \rangle \gets $ $(1+\epsp)$-approximate truss order of $G(x)$ \quad // \textit{\whp by Lemma~\ref{lem:cost-order}} \label{step:gx}
    \item if the first spurious edge $e_s$ appears \emph{before} the last edge $e_g$ of $G$, then \label{ln:condition}
    \vspace*{-.5em}
    \begin{enumerate}
        \setlength\itemsep{-.2em}
        \item $\tilde{t}_G \gets x$ \label{ln:updt}
        \item $x \gets \lceil (1+\epsilon)x \rceil$
        \item goto step~\ref{step:gx}
    \end{enumerate}
    \item if $\tilde{t}_G<2$ then return $0$  \quad // \textit{exact value of $t_{\gin}$} \label{step:return-zero}
    \item if a single $r \in \left[\tilde{t}_G \, (1+\epsp)^{-1}, \, (\tilde{t}_G+1)(1+3\epsp)\right]$ is a multiple of $6$, then return $\frac{r}{6}$  \quad // \textit{exact value}\label{step:return-exact}
    \item return $\frac{\tilde{t}_G}{6}$ \quad // $(1\pm \epsilon)$-\textit{approximation $\tilde{t}_{\gin}$} \label{step:return}
\end{enumerate}
\end{algorithm}

\subsection{Approximation algorithm}
\label{sec:estimate-trussness}
Suppose that we have a $(1+\epsilon)$-approximate truss order. We show that using the gadgets described in Section~\ref{sec:tools}, we obtain a $(1\pm\epsilon)$-approximation of $t_G$.
The basic idea is to expand the input graph using suitably $\gmirr{6}$ and then obtain $G(x)$; after that, the approximate order on $G(x)$ where $x$ increases each time by a factor $(1+\epsilon)$ is employed to spot the spurious cliques as ``markers'' to guess $t_G$. 
We give the pseudocode in Algorithm~\ref{alg:estimate}, whose rationale is explained in Lemma~\ref{lem:estimate}. The algorithm relies on the crucial test performed at line~\ref{ln:condition} that uses spurious edges as markers, according to what claimed in  Lemma~\ref{lem:gadgets}.

\begin{lemma}
    Given a graph $G$ and an integer $x\ge 0$, let $\langle e_1,\ldots, e_{|E(G(x))|}\rangle$ be a $(1+\epsilon)$-approximate truss order of $G(x)$.
    Let $e_g$ be the \emph{last} edge of $G$ appearing in the order, and $e_s$ the \emph{first} spurious edge appearing in the order. Then:
    \begin{enumerate}
        \item if $x<(1+\epsilon)^{-1}t_G$ then $e_g$ appears after $e_s$ 
        \item if $x>(1+\epsilon)\, t_G$ then $e_g$ appears before $e_s$ 
    \end{enumerate}
    \label{lem:gadgets}
\end{lemma}
\begin{proof}
By~\eqref{eq:approx-order} for $G(x)$, each edge $e_i$ has $\supp_{G(x)_{\ge e_i}}(e_i) \le \max\{\frac{T'}{m'}, (1+\epsilon)\,\minsupp(G(x)_{\ge e_i}) \}$, where  $T'$ the number of triangles in $G(x)$ and $m'$ is the number of edges.
From Lemma~\ref{lem:spurious-cliques} we also have $\frac{T'}{m'} \le \min\{\frac{t_G}{2}+\frac{x}{3}, \max\{t_G, \frac{x}{3}\}\}$.

We now analyze the two cases in the statement.
\begin{enumerate}
\item Case $x<(1+\epsilon)^{-1}t_G$: First observe that 
\begin{equation}
\label{eq:apx-support1}
  \max\left\{ \frac{t_G}{2}+\frac{x}{3}, \, (1+\epsilon)\, x\right\} < t_G
\end{equation}
by applying $x<t_G$ to the first argument of $\max$, and $x(1+\epsilon) < t_G$ to the second argument. Consider now an arbitrary edge $e_i$, recalling that 
\begin{equation}
\label{eq:apx-support2}
  \supp_{G(x)_{\ge e_i}}(e_i) \le \max\left\{ \frac{t_G}{2}+\frac{x}{3}, \, (1+\epsilon)\,\minsupp(G(x)_{\ge e_i})\right\} 
\end{equation}
holds by definition of approximate truss order, where $\frac{T'}{m'} \le \frac{t_G}{2}+\frac{x}{3}$  from Lemma~\ref{lem:spurious-cliques}. If $G(x)_{\ge e_i}$ contains a spurious edge, as the latter has support at most $x$, we have $\minsupp(G(x)_{\ge e_i}) \leq x$ in~\eqref{eq:apx-support2} as each spurious clique is isolated in $G(x)$. Hence we can upper bound~\eqref{eq:apx-support2} using this and~\eqref{eq:apx-support1}, obtaining
\begin{equation}
\label{eq:apx-support3}
  \supp_{G(x)_{\ge e_i}}(e_i)_ \le \max\left\{ \frac{t_G}{2}+\frac{x}{3}, \, (1+\epsilon)\,x\right\}  < t_G
\end{equation}
In other words, $e_i$ cannot be the first edge in the order of a $t_G$-truss because its forward triangles would be at least $t_G$.
Consequently, all the spurious edges (including $e_s$) occur before the first edge of a $t_G$-truss, and thus before $e_g$.

\item Case $x > (1 + \epsilon) \, t_G$: 
To prove this case we use a similar analysis to the one above, observing that $x > \max\{ \max\{t_G,\frac{x}{3}\}, (1+\epsilon)\,t_G\}$, since $x > t_G$. 

We use the fact $\frac{T'}{m'} \le \max\{t_G,\frac{x}{3}\}$ from Lemma~\ref{lem:spurious-cliques}. Hence, in the approximate truss order 
\begin{equation}\label{eq:apx-support4}
\supp_{G(x)_{\ge e_i}}(e_i) \le \max\left\{ \max\left\{t_G,\frac{x}{3}\right\}, (1+\epsilon)\,\minsupp(G(x)_{\ge e_i}) \right\}
\end{equation}
If $G(x)_{\ge e_i}$ contains an edge from $G$, we have $\minsupp(G(x)_{\ge e_i})\le t_G$ (by definition of trussness as the minimum support for an edge in any subgraph of $G$ cannot be larger than $t_G$). Hence we can upper bound~\eqref{eq:apx-support4} as 
\begin{equation}\label{eq:apx-support5}
\supp_{G(x)_{\ge e_i}}(e_i) \le \max\left\{ \max\left\{t_G,\frac{x}{3}\right\}, (1+\epsilon)\,t_G \right\} < x
\end{equation}
It follows that $e_i$ cannot be the first spurious edge in the ordering, which would have $x$ forward triangles as each spurious clique in $G(x)$ is isolated. Thus $e_s$ must occur in the order after $e_g$, which proves the statement for this case. (To get a full picture, all the spurious edges are at the end of this order.)
\end{enumerate}
\end{proof}

\begin{lemma}
\label{lem:estimate}
Consider any undirected graph $G$ with $m$ edges, $T_G$ triangles, arboricity $\alpha_G$, and trussness $t_G$. 
For any $0<\epsilon < 1$, suppose that a $(1+\epsilon)$-approximate truss order can be found in expected $f(\epsilon, m, T_G, \alpha_G) = O\left(\epsilon^{-2} \min\left\{\frac{m \log m}{T_G+1}, 1\right\}\, m \, \alpha_G\right)$ time and with probability $\Omega(1 - m^{-2})$, using $s(\epsilon, m) = O(\epsilon^{-2}m \log m)$ space.
Then, Algorithm~\ref{alg:estimate} provides \whp a $(1\pm \epsilon)$-approximation for the trussness $t_G$ in expected $O\left(\epsilon^{-1} \log (t_G+2) \cdot f(\epsilon, m, T_G, \alpha_G) \right)$ time, using $O(m+ s(\epsilon,m))$ space.
\end{lemma}
\begin{proof}
We replace the input graph $\gin$ with $G$, the graph obtained by uniting $\gmirr{6}$ and a $K_3$ (line~\ref{step:trasformation}): in this way, we guarantee that $t_G \geq 1$ and that we can take $\epsp = \frac{\epsilon}{6}$. It is not difficult to see at this point that a $(1\pm \epsilon)$-approximation of the trussness of the original input graph follows by dividing by~$6$ (line~\ref{step:return}) the $(1\pm \epsp)$-approximation of the trussness of the new $G$, as $\gmirr{6}$ amplifies trussness by $6$ (Lemma~\ref{lem:amplify-trussnes}), unless $t_{\gin}$ is $0$, in which case the algorithm will output the correct result (that is clearly also a $(1\pm \epsilon)$-approximation).

The algorithm performs some initialization steps (Line~\ref{step:init}), 
then outputs the approximation of the trussness of $\gin$ based on the value of $\tilde{t}$, which is computed in Lines~\ref{step:gx}--\ref{ln:condition} and bounded as follows:

\begin{enumerate}

\item \label{point:1} We prove that $\tilde{t}_G \le (1+\epsp)t_G \le (1+\epsilon)t_G$.

As by assumption $t_G\ge 1$ (because of the added $K_3$), in the beginning $\tilde{t}_G = 1 < (1+\epsp)t_G $.
Furthermore, if in a step we have $x > (1+\epsp)t_G$, then by Lemma~\ref{lem:gadgets} the first spurious edge $e_s$ appears after the last edge $e_g$ of $G$, meaning that $\tilde{t}_G$ will not be updated, so whenever we update $\tilde{t}_G$ we maintain that $\tilde{t}_G \le (1+\epsp)t_G$.

\item \label{point:2} We prove that $\tilde{t}_G \ge (1+3\epsp)^{-1} t_G  -1 \ge (1-\epsilon) t_G -1$.

Note that $\tilde{t}_G$ is the largest value of $x$ for which, in the $(1+\epsp)$-approximate truss order of $G(x)$, the condition at Line~\ref{ln:condition} is true.

By Lemma~\ref{lem:gadgets}, if $x < (1+\epsp)^{-1} t_G$ then the condition at Line~\ref{ln:condition} is true. Hence, let $x_F$ be first value of $x$ for which the condition is false (i.e., $e_s$ appears after $e_g$), and  let $x_T$ be the value of $x$ soon before $x_F$. It must be $x_F \ge \frac{t_G}{1+\epsp}$ and $\tilde{t}_G = x_T$. By looking at step~\ref{ln:updt}, we have $x_F = \lceil x_T (1+\epsp)\rceil = \lceil \tilde{t}_G (1+\epsp)\rceil$.
Thus $\lceil \tilde{t}_G (1+\epsp)\rceil \ge \frac{t_G}{1+\epsp}$,
meaning that $\tilde{t}_G (1+\epsp) +1 \ge \frac{t_G}{1+\epsp}$, and $\tilde{t}_G \ge \frac{t_G}{(1+\epsp)^2} -\frac{1}{1+\epsp}\ge \frac{t_G}{(1+\epsp)^2} -1 \ge \frac{t_G}{(1+3\epsp)} -1 \ge t_G(1-3\epsp) -1 \ge t_G(1-\epsilon)-1$.\footnote{Note that $1+3\epsp > (1+\epsp)^2$ and $\frac{1}{1+3\epsp} > 1-3\epsp$}

\end{enumerate}

We observe that our $(1 \pm \epsilon)$-approximation still has an off-by-1 in point~\ref{point:2} above: we get $\tilde{t}_G \ge (1-\epsilon)t_G-1$ instead of 
$\tilde{t}_G \ge (1-\epsilon)t_G$. The latter can be obtained with some extra math, requiring also the check at line~\ref{step:return-exact}, as discussed in the two cases below.
\begin{itemize}
\item Case $t_G\ge \frac{1}{3\epsp}$:
Since $3\epsp{t_G} \ge 1$, we get $\tilde{t}_G \ge (1+3\epsp)^{-1} t_G - 1 \ge t_G(1-3\epsp) - 1 \ge t_G(1-3\epsp) - 3\epsp{t_G} = t_G(1-6\epsp)$, proving the statement as $\epsilon = 6\epsp$.

\item Case $t_G < \frac{1}{3\epsp}$:
Recall that $t_G \geq (1+\epsp)^{-1}\tilde{t}_G$ by point~\ref{point:1}
and that
$t_G \leq (\tilde{t}_G + 1)  (1+3\epsp)$ by point~\ref{point:2}.
We want to prove that $(\tilde{t}_G+1)(1+3\epsp)- (1+\epsp)^{-1}\tilde{t}_G \le 4$ holds, thus ensuring that the algorithm will output the correct value of $t_G$ as it is the only multiple of $6$ in the given range (check done at line~\ref{step:return-exact}).
We have
$(\tilde{t}_G+1)(1+3\epsp)-(1+\epsp)^{-1}\tilde{t}_G \le
\tilde{t}_G\left(1+3\epsp-\frac{1}{1+\epsp}\right) + 1 + 3\epsp \le
\tilde{t}_G\left(\frac{\epsp}{1+\epsp}+3\epsp\right) + 2 \le
t_G\left(\epsp+3\epsp(1+\epsp)\right) + 2 \le
t_G(4\epsp+3\epsp^2) + 2 \le 5\epsp t_G+2 \le 4$, where we used the bounds on $\epsp$ which imply $3\epsp^2 \le \epsp$ and the upper bound $t_{G}(1+\epsp)$ on $\tilde{t}_G$.
\end{itemize}

As for the correctness,  we compute the approximate truss order $O(\frac{\log (t_G+2)}{\log 1+\epsilon}) = O(\epsilon^{-1}\log (t_G+2))$ times. Our algorithm still succeeds \whp as the latter probability is $\Omega(1-m^{-2})$ and we can generously bound $\epsilon^{-1}\log (t_G+2)$ as $O(n\log n)$.

As for the complexity, at each iteration we pay $O( f(\epsilon, m, T_G, \alpha_G))$. As the number of iterations is $O(\epsilon^{-1}\log (t_G+2))$, and we use $O(m)$ extra space, the statement follows.
\end{proof}

In order to conclude the proof of Theorem~\ref{thm:estimate}, we need to show in Section~\ref{sec:computix-apx-truss-order} how to compute a $(1+\epsilon)$-approximate truss order in expected $f(\epsilon, m, T_G, \alpha_G) = O\left(\epsilon^{-2} \min\left\{\frac{m \log m}{T_G+1}, 1\right\}\, m \, \alpha_G\right)$ time and with probability $\Omega(1 - m^{-2})$, using $s(\epsilon, m) = O(\epsilon^{-2}m \log m)$ space.

\subsection{Computing a $(1+\epsilon)$-approximate truss order}
\label{sec:computix-apx-truss-order}
\label{sec:gtri}

Given a graph $G$, we denote as $\gtri$ the \textit{triangle hypergraph}. This hypergraph, originally introduced in~\cite{burkhardt2018bounds}, is a $3$-uniform hypergraph whose nodes are the edges of $G$, and whose hyperedges are the triplets of edges which form a triangle in $G$. More formally, $V(\gtri) = \{v_{e_i} : e_i $ is an edge of $G\}$, and for each triplet $e_1, e_2, e_3\in E(G)$ which forms a triangle, there is a hyperedge $(v_{e_1},v_{e_2},v_{e_3}) \in E(\gtri)$.
Since each pair of edges defines a unique triangle (if any), each pair of hyperedges in $\gtri$ can overlap by at most one vertex, and the degree of $v_{e_i}$ in $\gtri$ is equal to the support of $e_i$ in $G$. Thus, $\gtri$ has $|E(G)| = m$ vertices and $T_G$ hyperedges. Moreover, building $\gtri$ takes $O(md_G)$ time.

We will use $\gtri$ for computing the approximate truss order, since there is correspondence between the trussness of $G$ and the degeneracy of $\gtri$. For each node $u$ of a hypergraph $H$ its degree, indicated as $\degr_H(u)$, is the number of hyperedges in $H$ containing $u$. The degeneracy of the hypergraph $\gtri$, denoted as $d_{\gtri}$, is the natural extension of degeneracy for graphs: it is the maximum $d$ such that there is an induced subgraph $H'$ of $H$ in which, for each vertex $u$, $\degr_{H'}(u)\geq d$.

Given the correspondence between degree in $\gtri$ and support in $G$, as noted in~\cite{burkhardt2018bounds}, it is straightforward to observe that a $k$-core in $\gtri$ (a subgraph where all nodes have degree $k$ or more) corresponds to a $k$-truss in $G$ (a subgraph where all edges have support $k$ or more). Hence we obtain:

\begin{remark}
$d_{\gtri} = t_G$
\label{rem:dgtri}
\end{remark}

Our method for computing an approximate truss order consists in efficiently computing a sample of $\gtri$, as explained in Section~\ref{sec:sampling}, then using that sample to compute the approximate truss order, as shown in Section~\ref{sec:degtrip}.

\subsubsection{Sampling}
\label{sec:sampling}

In the following, we show how to sample a subhypergraph $\gtrip$ of $\gtri$ such that each hyperedge of $\gtri$ is in $\gtrip$ independently with probability $p \ge \zeta\frac{m\log m}{(T_G+1)\epsilon^{2}}$ 
, where $\zeta$ is a suitable constant such that Lemma~\ref{lem:dege} holds, in less than $O(m \, \alpha_G)$. 

Note that we must do this \textit{without} access to $\gtri$ or even the value of $p$: A simple procedure to obtain $\gtrip$ would scan the hyperedges of $\gtri$, i.e. the probability space, retaining them with probability $p$. This requires knowledge of $p$ and scanning all the triangles of $G$, with $O(m\,\alpha_G)$ time cost. As we cannot afford this cost, we use an alternative equivalent strategy for generating $\gtrip$. 

Sampling uniformly at random hyperedges of $\gtri$ corresponds to uniformly sampling triangles of $G$: there is a 1-to-1 correspondence between triangles in $G$ and edges in $\gtri$, thus our probability space is the set of triangles of $G$, each one corresponding to a different hyperedge of $\gtri$. As obtaining the $i$th triangle of $G$ is not easy, we use a different strategy to sample them, based on wedges. In particular, we observe that if the nodes are ordered, each triangle corresponds to a unique \textit{forward wedge} (i.e., a wedge in which the middle node is the smallest of the three). For this reason, we sample forward wedges, which can be handled more easily, rather than triangles, and we accept them, adding the corresponding hyperdedge of $\gtri$ to $\gtrip$ only if they are closed.

More precisely, let $W$ be the number of forward edges in $G$. Our sampling procedure works as follows.

\begin{enumerate}
\item\label{enum:1}
Start setting $p=\zeta \frac{m\log m}{W\epsilon^{2}}$.
\item\label{enum:2} Extract each forward wedge with probability $p$.
\item\label{enum:3}  For each forward wedge $w$ sampled, if $w$ is closed, i.e. it corresponds to a triangle, add to $\gtrip$ the hyperedge of $\gtri$ which corresponds to $w$.
\item\label{enum:4} If the number of hyperedges of $\gtrip$ is not at least $\frac32\zeta\frac{m\log m}{\epsilon^{2}}$, double $p$. If $p\geq 1$ run the exact algorithm, else clear $\gtrip$ and go to step~\ref{enum:2}.
\end{enumerate}

First of all, $W$ can be computed in $O(m)$ time, using a degeneracy ordering of the graph: Let $\degr(u)$ be the \textit{forward degrees} of a node $u$ in $G$, i.e. the number of edges going forward in the order from $u$. As each distinct pair of forward neighbors of $u$ makes with $u$ a distinct forward wedge, we have $W=\sum_{u\in V(G)} \degr(u)(\degr(u)-1)/2$.

In step~\ref{enum:2}, in order avoid paying $O(W)$ time to choose which forward edges should be kept, we use Method 9 in~\cite{fan1962development} (as it is often done to generate $G(n,p)$ Erd\H{o}s-R\'{e}nyi graphs), which allows us to obtain the same probability distribution by only paying the cost to sample $k$ geometrically distributed variables with parameter $p$, where $k$ is the number of sampled wedges. To generate a geometric random variable we use the method in~\cite{bringmann2013exact} whose expected cost is $O(1 + \log(1/p)/w)$, where $w$ is $\Omega(\log n)$, and in our case is constant as $1/p$ is $O(\frac{n}{\epsilon^2})$ and $\epsilon \ge 1/n$.\footnote{As $t_G \le n$, $\epsilon<1/n$ does not give any more information than setting $\epsilon = 1/n$.}

We prove that step~\ref{enum:2} and~\ref{enum:3} sample triangles independently at random with probability $p$. 

\begin{lemma}
Step~\ref{enum:2} and~\ref{enum:3} of our sampling process to obtain $\gtrip$ are equivalent to scanning each hyperedge of $\gtri$ and accepting it with probability $p$.%\prpshort
\label{lem:wedge}
\end{lemma}
\begin{proof}

Step~\ref{enum:3} adds the hyperedge corresponding to a sampled wedge $w$ only if $w$ is closed. Scanning all the pairs $w$ and sampling each of them with probability $p$ implies that also each closed wedge $w$ is sampled with probability $p$. Since in $\gtri$ each hyperedge is involved in just one triangle and hence it corresponds to a unique forward wedge $w$, we obtain that a hyperedge of $\gtri$ is also added to $\gtrip$ with probability $p$.
\end{proof}

The sampling procedure repeats step~\ref{enum:2} and step~\ref{enum:3}, performing doubling on $p$, and ends up sampling triangles with probability at least $\zeta\frac{m\log m}{(T_G+1)\epsilon^{2}}$ \whp as shown next.

\begin{lemma}\label{lem:sampl-prob}
Our sampling process of $\gtrip$ samples edges of $\gtri$ independently and uniformly at random with probability $\zeta \frac{m\log m}{(T_G+1)\epsilon^{2}} \le p \le 4\zeta \frac{m\log m}{(T_G+1)\epsilon^{2}}$, \whp.%\prpshort
\end{lemma}
\begin{proof}
From Lemma~\ref{lem:wedge} we know that $\gtrip$ is an uniform and independent sample of $\gtri$ with probability $p$. Thus, we only have to prove that when the algorithm terminates we have that $p$ is bounded as in the statement.

To prove the lower bound on $p$, we show that if $p\le \zeta \frac{m\log m}{(T_G+1)\epsilon^{2}}$ then the number of edges in $\gtrip$ is less than $\frac{3}{2}\zeta\frac{m\log m}{\epsilon^2}$ \whp. We use a well-known application of the Chernoff bounds on the tail of a binomial distribution, i.e. that for a binomial variable $X\sim B(T_G, p)$ (where $X$ represents the number of edges in $\gtrip$) we have $\mathbb{P}\left\{X \ge (1+\delta) T_Gp\right\} \le \exp\left(-\frac{\delta^2}{3}T_Gp\right)$.

Setting $\delta = \frac{3}{2}\zeta \frac{m \log m}{T_Gp\epsilon^2}-1$, which trivial calculations show that is at least $\frac{1}{2}\zeta\frac{m \log m}{T_Gp\epsilon^2}$ if we assume $p\le \zeta\frac{m \log m}{(T_G+1)\epsilon^{2}}$, we obtain

\[
\begin{array}{rcl}
\mathbb{P}\left\{X \ge \frac{3}{2}\zeta \frac{m \log m}{\epsilon^2}\right\} & \le & \exp\left(-\frac{1}{12}\zeta^2 \frac{m^2 \log^2 m}{T_Gp\epsilon^4}\right) \\
& \le & \exp\left(-\frac{1}{12}\zeta \frac{m \log m}{\epsilon^2}\right)
\end{array}
\]

which is exponentially small in $m$ and $\epsilon^{-1}$. 

We can obtain a similar upper bound on $\mathbb{P}\left\{X \le \frac{3}{2}\zeta \frac{m \log m}{\epsilon^2}\right\}$ when $p \ge 2\zeta\frac{m \log m}{(T_G+1)\epsilon^{2}}$ using the Chernoff bound on the lower tail of the distribution of $X$.
Thus, \whp we have that the stopping condition holds if $p \ge 2\zeta\frac{m \log m}{(T_G+1)\epsilon^{2}}$ and does not hold if $p < \zeta\frac{m \log m}{(T_G+1)\epsilon^{2}}$. 
Since we double $p$ at each step, $p$ is twice the value $p'$ it had in the last step where the stopping condition did not hold. This means that, \whp, $p'<2\zeta\frac{m \log m}{(T_G+1)\epsilon^{2}}$ and thus $p < 4\zeta\frac{m \log m}{(T_G+1)\epsilon^{2}}$. On the other hand, the stopping condition does not hold when $p < \zeta\frac{m \log m}{(T_G+1)\epsilon^{2}}$, \whp, thus \whp $p$ must be at least $\zeta\frac{m \log m}{(T_G+1)\epsilon^{2}}$, proving the statement.
\end{proof}

Which means the size of $\gtrip$ is at most $4\zeta\frac{m\log m}{\epsilon^{2}}$ in expectation. Finally, the complexity of the process is as follows (we recall $d_G \le 2\cdot \alpha_G$):

\begin{lemma}
Sampling $\gtrip$ costs $O\left(\min\left\{\frac{m \log m}{(T_G+1)\epsilon^{2}}, 1\right\}md_G\right)$ time in expectation.%\prpshort
\label{lem:cost}
\end{lemma}
\begin{proof}
Since the expected cost of extracting a geometrically distributed random variable is constant, the expected running time of the sampling in step~2 is the expected number of sampled wedges, i.e. $O(pW)$. Since at each iteration $p$ doubles and $p$ \whp is at most $4\zeta\frac{m \log m}{(T_G+1)\epsilon^{2}}$, the last iteration dominates the running time. If $p\geq 1$ the same happens with the exact algorithm. Hence, we obtain the claimed running time, noticing that $W = O(md_G)$ if the graph is in degeneracy order. Note that $\min\left\{\frac{m \log m}{(T_G+1)\epsilon^{2}}, 1\right\}md_G \ge m$, meaning that the cost of computing a degeneracy order is dominated by the running time of the rest of the algorithm.
\end{proof}

\subsubsection{Approximate truss order}
\label{sec:degtrip}

By Remark~\ref{rem:dgtri}, the degeneracy of $\gtri$ corresponds to the trussness of $G$. In this section we aim to compute an approximate truss order of $G$ by computing an approximate degeneracy ordering of $\gtri$. This can be done using $\gtrip$, and extending some of the techniques in~\cite{farach2016latin} to hypergraphs.

Given a hypergraph $H$ and a subset of its nodes $U$, we indicate with $H[U]$ the subhypergraph induced by the nodes in $U$, i.e. the hypergraph having as nodes $U$ and as hyperedges the ones of $H$ subset of $U$. 

\newcommand{\mindegr}{\degr_{min}\xspace}

\begin{lemma}
Let $H$ be a hypergraph with $n$ nodes and $m$ hyperedges, and $H_p$ be the hypergraph obtained from $H$ by retaining each hyperedge independently with probability $p = \zeta\frac{n\log n}{\epsilon^2m}$ for any constant $\zeta > 54$. Moreover, let $U$ be a subset of the nodes of $H$ and $v$ be the smallest-degree node in $H_p[U]$. Then, with probability at least $1-n^{1-\frac{\zeta}{27}}$, the following holds:
$\degr_{H[U]}(v) \le \max\left\{\frac{m}{n}, (1+\epsilon)\mindegr(H[U])\right\}$.
\label{lem:mindeg}
\end{lemma}
\begin{proof}
Let $\zeta = 27\gamma$, so that $\gamma \ge 2$. We prove that, for any two nodes $x$ and $y$ with $\degr_{H[U]}(x) \ge \max\left\{\frac{m}{n}, (1+\epsilon)\degr_{H[U]}(y)\right\}$, we have $\mathbb{P}\left\{\degr_{H_p[U]}(x) \le \degr_{H_p[U]}(y)\right\} \le \frac{1}{n^{\gamma}}$. This implies that the thesis holds \whp, as the probability that the thesis fails is bounded by the probability that at least one node $x$ in $U$ has a degree lower than the one of $v$ in $H_p[U]$ while satisfying $\degr_{H[U]}(x) \ge \max\left\{\frac{m}{n}, (1+\epsilon)\degr_{H[U]}(v)\right\}$. This probability is bounded by $n\cdot \frac{1}{n^{\gamma}}=\frac{1}{n^{\gamma-1}}$, which means that the statement holds.

Note that we only need to consider the case when $\delta_{H[U]}(y)\ge \frac{1}{1+\epsilon}\frac{m}{n}$. 
Indeed, 
when $\delta_{H[U]}(y)< \frac{1}{1+\epsilon}\frac{m}{n}$, we have
$\max\left\{\frac{m}{n}, (1+\epsilon)\degr_{H[U]}(y)\right\}=\frac{m}{n}$. Moreover, if we replace the nodes $y$ such that $\delta_{H[U]}(y)< \frac{1}{1+\epsilon}\frac{m}{n}$ with nodes of degree equal to $\frac{1}{1+\epsilon}\frac{m}{n}$, the set of possible values of $x$ does not change, and for a fixed $x$ the probability $\mathbb{P}\left\{\degr_{H_p[U]}(x) \le \degr_{H_p[U]}(y)\right\}$ increases when $\degr_{H_p[U]}(y)$ increases, meaning that bounding this probability in the case $\delta_{H[U]}(y)\ge \frac{1}{1+\epsilon}\frac{m}{n}$ implies a bound for the case $\delta_{H[U]}(y)< \frac{1}{1+\epsilon}\frac{m}{n}$.

Thus, as $\degr_{H[U]}(x) \ge \max\left\{\frac{m}{n}, (1+\epsilon)\degr_{H[U]}(y)\right\}$, we can assume that $\delta_{H[U]}(y)\ge \frac{1}{1+\epsilon}\frac{m}{n}$, $\delta_{H[U]}(x)\ge \frac{m}{n}$, and $\degr_{H[U]}(x) \ge (1+\epsilon)\degr_{H[U]}(y)$. Defining $c=\frac{\epsilon}{2+\epsilon}$, so that $(1-c)\degr_{H[U]}(x) \ge (1+c)\degr_{H[U]}(y)$, we have that the event $E_1=(\degr_{H_p[U]}(x) \le \degr_{H_p[U]}(y))$ implies the event $E_2=(\degr_{H_p[U]}(x) \le (1-c)p\degr_{H[U]}(x) \vee \degr_{H_p[U]}(y) \ge (1+c)p\degr_{H[U]}(y))$, as $\lnot E_2$ implies $\lnot E_1$ because of the choice of $c$.
Hence, we have:

\[
\begin{array}{rcl}
    \mathbb{P}\left\{\degr_{H_p[U]}(x) \le \degr_{H_p[U]}(y)\right\} & \le & \mathbb{P}\left\{\degr_{H_p[U]}(x) \le (1-c)p\degr_{H[U]}(x) \vee \degr_{H_p[U]}(y) \ge (1+c)p\degr_{H[U]}(y)\right\} \\
    & \le & \exp\left(-\frac{c^2}{2}p\degr_{H[U]}(x)\right) + \exp\left(-\frac{c^2}{c+2}p\degr_{H[U]}(y)\right) \\
    & \le & 2\exp\left(-\gamma\log n\right) \\
    & \le & \frac{1}{n^{\gamma}}
\end{array}
\]
where the second inequality follows from the fact that the degree distribution of a vertex in $H_p$ is binomial, as well as the Chernoff bounds, and the third one from the fact that $c^2p\degr_{H[U]}(y) \ge c^2p\frac{m}{n} \ge 3\gamma\log n$ because of our hypotesis on $p$, and similarly for $c^2p\degr_{H[U]}(x)$.
\end{proof}

In the following, we obtain the main property of our approximate truss order.

\begin{lemma}
Given a graph $G$ with $n$ nodes and $m$ edges, let $v_{e_1},\ldots, v_{e_m}$, be the order of the edges of $E(G)$ corresponding to a degeneracy ordering the nodes of $\gtrip$, for $p = \zeta\frac{m\log m}{\epsilon^2T_G}$ for any constant $\zeta>81$. Then the following holds on $G$ with probability at least $1-m^{2-\frac{\zeta}{27}}$: the forward triangles of $e_i$ in $G$ are at most $\max\left\{ \frac{T_G}{m}, (1+\epsilon)s\right\}$, where $s$ is the smallest support of an edge in $G_{\geq e_i}$.\footnote{To guarantee a failure probability that is $o(m^{-2})$, as required by Lemma~\ref{lem:estimate}, we need $\zeta > 108$.}
\label{lem:dege}
\end{lemma}
\begin{proof}
By construction (see Section~\ref{sec:sampling}), we have that each edge of $\gtri$ appears in $\gtrip$ with probability $\zeta\frac{m\log m}{(T_G+1)\epsilon^{2}}$, independently from the others. Since a degeneracy ordering is computed by iteratively removing the lowest-degree vertex from $\gtrip$, applying $m$ times Lemma~\ref{lem:mindeg} we get that $v_{e_1},\ldots, v_{e_m}$ is, with probability at least $1-m^{2-\frac{\zeta}{27}}$, an approximate degeneracy ordering of $\gtri$.

This means that the forward neighbors of each $v_{e_i}$ are bounded by $\max\left\{\frac{E(\gtri)}{V(\gtri)} , (1+\epsilon)\degr_{min}\right\}$ where $\degr_{min}$ is the smallest degree of a node in $\gtri_{\geq v_{e_i}}$.
By construction of $\gtri$, the corresponding order $e_1,\ldots, e_m$ on the edges of $G$ has the property that the forward \emph{triangles} of each $e_i$ are bounded by $\max\left\{\frac{E(\gtri)}{V(\gtri)} , (1+\epsilon)s\right\}$ where $s$ is the smallest support of an edge in $G_{\geq e_i}$. Thus, the statement follows.
\end{proof}

As a result, the cost of the getting a $(1+\epsilon)$-approximate truss order follows. We can finally prove the following statement (corresponding to that of Theorem~\ref{thm:estimate}).

\begin{lemma}
There is an algorithm that computes \whp a $(1+\epsilon)$-approximate truss order of $G$ in expected time $O\left(\min\left\{\frac{m \log m}{(T_G+1)\epsilon^{2}}, 1\right\}md_G\right)$ and space $O(\epsilon^{-2}m \log m)$.
\label{lem:cost-order}
\end{lemma}
\begin{proof}
By Lemma~\ref{lem:cost}, $\gtrip$ can be obtained with cost $O\left(\min\left\{\frac{m \log m}{(T_G+1)\epsilon^{2}}, 1\right\}md_G\right)$ and has  $O(\epsilon^{-2}m \log m)$ hyperedges. As the cost of computing a degeneracy order is bounded by the number of hyperedges, the statement follows, since if we end up running the exact algorithm we have a total running time of $O(md_G)$, and if not then $\frac{m \log m}{(T_G+1)\epsilon^{2}}md_G \ge \epsilon^{-2}m \log m$, i.e. the cost of sampling $\gtrip$ dominates over the cost of computing its degeneracy order.
\end{proof}

\section{Conditional Lower Bounds}
\label{sub:lower-bound-T-approx}

The trivial algorithm for counting triangles in a graph takes $O(n^3)$ time: determining whether this bound is inherent or it can be improved to a truly subcubic one (e.g., $O(n^{2.9})$) by a ``combinatorial'' algorithm has been a long standing question, even referred to as a ``holy grail'' of graph algorithms~\cite{williams2010subcubic}.

What is a ``combinatorial'' algorithm is hard to define formally, but the term aims at designating practical approaches which has reasonable constant factors~\cite{williams2010subcubic,yu2015improved,aingworth1999fast}. As a rule-of-the-thumb, combinatorial algorithms are not based on fast matrix multiplication, which has better asymptotic complexity (i.e., $O(n^\omega)$ for some $\omega < 2.373$ to multiply two $n\times n$ matrices) but very large hidden constant factors due to a high number of subproblems generated in the matrix multiplication that makes it not practically efficient.

This conditional lower bound for combinatorial algorithms is further strengthened in~\cite{williams2010subcubic}, which shows its equivalence to other important problems, such as Boolean Matrix Multiplication (BMM):

\begin{theorem}[\normalfont{(from \cite{williams2010subcubic})}]
\label{thm:subcubic}
The following problems either all have truly subcubic combinatorial algorithms, or none of them do:
\begin{itemize}
    \item Boolean matrix multiplication (BMM).
    \item Detecting if a graph has a triangle.
    \item Listing up to $n^{3-\delta}$ triangles in a graph for constant $\delta > 0$.
    \item Verifying the correctness of a matrix product over the Boolean semiring.
\end{itemize}
\end{theorem}

Conditional lower bounds hold also for the (multiplicative or additive) approximation of trussness if we want to design a combinatorial algorithm that takes $O(m (t_G+1))$ time in the worst case as shown by the following Lemma, whose proof is based on $\gmirr{q}$ introduced in Lemma~\ref{lem:amplify-trussnes}.
\begin{lemma}
\label{lem:conditional-tritrussmt-approx}
Given any undirected graph $G$ with $n$ nodes, $m$ edges, and trussness $t_G$, 
there is no combinatorial algorithm that, for any $\epsilon>0$, provides an approximation for the value of $t_G$ by either a multiplicative factor of $1 + \epsilon$ or an additive term of $O(n^{\frac{1}{2}-\epsilon})$, taking $O(m (t_G+1))$ time in the worst case, unless BMM is truly subcubic. %\prpshort
\end{lemma}
\begin{proof}
The arguments for the  multiplicative factor of $1 + \epsilon$ are analogous to those for Lemma~\ref{lem:conditional-tritrussmd}: if $t_G = 0$ then its approximate value is zero too and thus we can test if $G$ is triangle-free. Hence we focus on the additive term of $O(n^{\frac{1}{2}-\epsilon})$. To this aim, we need the gadget described in Lemma~\ref{lem:amplify-trussnes} to amplify the (unknown) trussness $t_G$ by any integer factor $q>1$. 
 Suppose that such an approximation algorithm $A$ exists for any $G'$ in $O(m (t_{G'}+1))$ time, and let it run on $G'=\gmirr{q}$ where $q = 3w$ and $G$ is the graph that we want to check for triangle-freeness. If so, its trussness is $t_G =0$, and thus $A$ on $\gmirr{q}$ returns an approximate value $\tilde{t}_\gmirr{q}  \leq t_\gmirr{q} + w  = q \, t_G + w = w$. Otherwise, $G$'s trussness is $t_G \geq 1$, and $A$ on $\gmirr{q}$ returns an approximate value $\tilde{t}_\gmirr{q}  \geq t_\gmirr{q} -  w \geq q - w > w$. In other words, checking if the output of $A$ is smaller or equal to $w$, we can tell whether $G$ is triangle-free. The time complexity of $A$ when $G$ is triangle-free is $O(|E(\gmirr{q})| \, (t_\gmirr{q}+1)) = O(|E(\gmirr{q})|) = O(m q^2) = O(m w^2) = O(n^2(n^{\frac{1}{2}-\epsilon})^2 ) = O(n^{3-2\epsilon})$ which improves BMM by Theorem~\ref{thm:subcubic} as $\epsilon>0$: indeed we can stop $A$ if it runs longer, and declare that $G$ is triangle-free if $A$ outputs an approximation of value $\leq w$.
\end{proof}

\begin{lemma}
Suppose that there exists a (combinatorial) algorithm to approximate the trussness of any graph $G$ containing $\Omega(m)$ triangle in $O(m^hn^k)$ time, within either a multiplicative factor $c$ or an additive term $\frac{c^2}{2}$, for $c \geq 1$. Then there is a (combinatorial) algorithm to recognize whether a graph $G$ is triangle free in $O(m^hn^kc^{4h+2k})$ time.%\prpshort
\label{lem:conditional-general}
\end{lemma}
\begin{proof}
Let $T_G = \Omega(m)$ be the number of triangles in $G$, recalling that $T_G=0$ if and only if the trussness is $t_G=0$. 
Given $G$ and its parameters:
\begin{itemize}
\item Consider the gadget $\gmirr{(c^2+2)}$ introduced in Lemma~\ref{lem:amplify-trussnes}, which has trussness $(c^2+2)t_G$ and $O(c^6 T_G)$ triangles. (Recall that we can navigate implicitly $\gmirr{(c^2+2)}$).
\item Let $K$ be a graph consisting of a complete bipartite graph $K_{\sqrt{c^4m},\sqrt{c^4m}}$, plus one node $v_x$ connected to all the others. $K$ has $O(c^2\sqrt{m}) = O(c^2n)$ nodes and $\Theta(c^4m)$ edges, and can be accessed implicitly similarly to $\gmirr{(c^2+2)}$. Furthermore, we can observe that it contains $\Theta(c^4m)$ triangles,
and its trussness is $t_K = 1$, since every edge not adjacent to $v_x$ belongs to exactly one triangle.
\end{itemize}
Now we take a new graph $G'$, which is simply the disjoint union of $\gmirr{(c^2+2)}$ and $K$. We observe that $G'$ has  $\Theta(cn + c^2\sqrt{m}) = O(c^2n)$ nodes, $\Theta(c^4m)$ edges, and $\Theta(c^6T_G)+\Theta(c^4m) = \Omega(c^4m)$ triangles.

Furthermore, if $t_G=0$, then the trussness of $G'$ is $t_{G'}=1$ because of $K$; otherwise, $G'$ has trussness $t_{G'} = (c^2+2)t \ge c^2+2$.

Let $A$ be an algorithm which can approximate the trussness either within a multiplicative factor of $c$ or within an additive factor of $\frac{c^2}{2}$, in $O(m^hn^k)$ time, provided that the graph has $\Omega(m)$ triangles. 
We observe that the assumption on the number of triangles is satisfied by $G'$. Thus, we can apply $A$ to $G'$ to get the approximate value $\tilde{t}_{G'}$, say within a multiplicative factor of~$c$. If $t_G=0$ then $\tilde{t}_{G'} \leq c \, t_{G'} = c$; else $\tilde{t}_{G'} \geq \frac{t_{G'}}{c} \geq \frac{c^2+2}{c} > c$.
When the approximation is within an additive term $\frac{c^2}{2}$, we have the following. If $t_G=0$ then $\tilde{t}_{G'} \leq t_{G'} + \frac{c^2}{2} = 1 + \frac{c^2}{2}$; else $\tilde{t}_{G'} \geq t_{G'} -  \frac{c^2}{2} \geq (c^2+2) - \frac{c^2}{2}> 1+\frac{c^2}{2}$. Hence
we can decide whether $G$ is triangle-free in $O(|V(G')|^h|E(G')|^k) = O( (c^4m)^h (c^2n)^k)) = O(m^h n^k c^{4h+2k})$ time.
\end{proof}

\begin{theorem}
\label{thm:conditional-bound-T-approx}
Consider any undirected graph $G$ with $n$ nodes, $m$ edges, arboricity $\alpha_G$, and trussness $t_G$. Even if its number of triangles is $\Omega(m)$ then there is no combinatorial algorithm that, for any given $c \geq 1$, provides an approximation of $t_G$ within a multiplicative factor of $c$ or an additive term of $c^2/2$, taking $\tilde{o}(m \, \alpha_G)$ time in the worst case, unless BMM is truly subcubic.
\end{theorem}
\begin{proof}
Follows from Theorem~\ref{thm:subcubic} and the more general statement of Lemma~\ref{lem:conditional-general}, recalling that $m \, \alpha_G = \Theta(n^3)$ in the worst case.
\end{proof}

Theorem~\ref{thm:conditional-bound-T-approx} says that we need more than order $m$ triangles in $G$ to hope to be significantly faster the $O(m \, \alpha_G)$ time bound for the exact computation.

\section{Approximating the Trussness (with Matrix Multiplication)}

For the sake of completeness, it is interesting to see what can be done if matrix multiplication is allowed. To do so we use Lemma~\ref{lem:support-gives-approx}, which enables us to approximate the trussness using triangle counting algorithms.

\begin{lemma}
Given an undirected graph $G$ with $n$ nodes, $m$ edges, and trussness $t_G$, suppose that the support for every edge of $G$ can be computed in $f(n,m) = \Omega(m)$ time and $S(n, m) = \Omega(m)$ space. For any $\epsilon > 0$, a $(3+\epsilon)$-approximation of $t_G$ can be computed in $O\left(\sum_{i=0}^{\log m / \log \epsilon} f\left(n, m\left( \frac{3}{3+\epsilon} \right)^i \right)\right) = \tilde{O}(f(n, m))$ time and $O(S(n, m))$ space.
\label{lem:support-gives-approx}
\end{lemma}
\begin{proof}
We follow the scheme in~\cite{Farach-ColtonT14}, and fix $c=3+\epsilon$. Let $A$ be the algorithm that computes the support for every edge of $G$. 
Let $G^1 = G$. For $i=1,2, \ldots$, run $A$ on $G^i$ and delete the edges with support less than or equal to $c\frac{T_i}{m_i}$, where $T_i$ is the number of triangles in $G^i$ and $m_i$ is its number of edges. We obtain $G^{i+1}$ in this way. Continue the iteration on $i := i+1$ until the current graph has no more edges. 
As the sum of the supports is three times the number of triangles, we have that $m_{i+1}\leq \frac{3T_i}{cT_i/m_i}=3\frac{m_i}{c}.$
In other words, at most a fraction  $\frac{3}{c} = \frac{3}{3+\epsilon} < 1$ of the previous edges survive at each iteration.

The returned approximate value $\tilde{t}_G$ is $\max_{i \geq 1} T_i/m_i$, for the values $T_i$ and $m_i$ seen in the whole process. 
Since all the edges have been removed, there exists a value of $i$ such that $t_G \leq c\frac{T_i}{m_i}$, as otherwise no edge would have been removed from the $t_G$-truss of $G$. Hence, $c\tilde{t_G} \geq t_G$. 

Applying Theorem~\ref{thm:nashwill}, we also have $t_G\geq T_i/m_i$ for each $i$, thus $t_G \geq \tilde{t_G}$. Hence, $\tilde{t_G}$ is a $(3+\epsilon)$-approximation.
\end{proof}

Using the previous lemma, we can finally prove Theorem~\ref{thm:approx-matrix-multiplication} 

\begin{theorem}
\label{thm:approx-matrix-multiplication}
Given an undirected graph $G$ with $n$ nodes, $m$ edges, and any $\epsilon > 0$, a $(2+\epsilon)$-approximation of the trussness $t_G$ can be computed in time $O(\epsilon^{-1}n^{\omega}\log{\frac mn})$ or $O(\epsilon^{-1}m^{1+\frac{\omega-1}{\omega+1}})$, where $O(n^\omega)$ is any upper bound for the $(n \times n)$-matrix multiplication cost.%\prpshort
\end{theorem}
\begin{proof}
We employ fast matrix multiplication for computing the support of each edge, with a cost of $O(\min\left\{n, m\right\}^{\omega})$ or $O(m^{1+\frac{\omega-1}{\omega+1}})$ time~\cite{le2014powers}, to the adiacency matrix of $G$.
Considering cell $(a, b)$, its value is given by the scalar product of the adjacency vectors of nodes $a$ and $b$. Hence, if $\{ a, b\}\in E(G)$, cell $(a, b)$ contains $\supp_G(\{a,b\})$. From this follows that we can compute the support of all edges in $f(n,m) = O(n^{\omega})$ or $f(n,m)  = O(m^{1+\frac{\omega-1}{\omega+1}})$ time by following the approach of~\cite{alon1997finding}.

We obtain the required running times by replacing the cost of support counting in the formulas provided by Lemma~\ref{lem:support-gives-approx}, observing that in the case of the cost of $O(\min\left\{n, m\right\}^{\omega})$ we can only consider the first $O\left(\log \frac{m}{n}\right)$ terms of the cost, as after that we obtain a geometric series which is bounded by the first term. In the second case, the geometric series starts from the first term, thus removing all logarithmic factors.
\end{proof}

We observe that it is possible to obtain a fast and practical algorithm from what described in the proof of Lemma~\ref{lem:support-gives-approx}, by replacing the matrix multiplication to compute the support of each edge, with a practical method to list triangles.

\bibliography{references}

\appendix

\section*{APPENDIX}

\section{Trussness and truss decomposition}
\label{app:lowerb}

 We use the following relations:
$$\sum_{\{u,v\} \in E(G)} \min\{\degr(u),\degr(v)\} \leq 2 m \, \alpha_G, \quad \alpha_G = O(\sqrt{m}), \quad\mbox{ and }\quad  \alpha_G \geq \frac{t_G+1}{2}$$

See~\cite{chiba1985arboricity} for the first two. For the third one, since the degree of each node $u$ in a \ktruss is $\delta(u) \geq k+1$, we have that an inclusion-maximal \ktruss for $k=t_G$ has at least $n' (t_G+1)/2$ edges, where $n'$ is the number of its nodes. As each forest cover less than $n'$ of the edges in the \ktruss, it yields $\alpha_G \geq \frac{t_G+1}{2}$

\begin{lemma}
\label{lem:tritrussmd}
Given an undirected graph $G$ with $m$ edges and arboricity $\alpha_G$, its triangles can be combinatorially counted or listed and its trussness can be computed in $O(m \, \alpha_G)$ time and $O(m)$ space.
\end{lemma}
\begin{proof}
Let $\degr(v) = |N_G(v)|$ denote the degree of a node $v$.
Consider the algorithm that scans all the edges in $E(G)$ and, for each edge $e=\{u,v\}$, it computes $\supp_G(e) = |N_G(u) \cap N_G(v)|$ in $O(\min\{\degr(u),\degr(v)\})$ time by taking the node of minimum degree between $e$'s endpoints, say $u$, and checking if its neighbors (in $N_G(u)$) are adjacent to the other endpoint, say $v$. This gives both triangle listing and counting, where the latter quantity is obtained as $T_G = \sum_{e \in E(G)} \supp_G(e)$. Total time is $O(m + \sum_{\{u,v\} \in E(G)} \min\{\degr(u),\degr(v)\})$ which is proved in~\cite{chiba1985arboricity} to be $O(m \, \alpha_G)$.

As for the trussness $t_G$, an extra postprocessing step is needed for the
algorithm. After computing the support of each edge, as mentioned above, it
sorts the edges $e \in E(G)$ in non-decreasing order with respect to
$\supp_G(e)$ in $O(m)$ time. It then keeps the edges in buckets corresponding to
their support, so that any two edges $e$ and $e'$ are in the same bucket iff
$\supp_G(e)=\supp_G(e')$. These buckets can be built in $O(m)$ time, and they
can be easily managed dynamically, thus taking the edge of smallest support and
changing the bucket of an edge (whose supports has to change), in $O(1)$ time.
As long as the buckets are nonempty, the algorithm removes an edge $e=\{u,v\}$
from the nonempty bucket with the smallest associated edge support, and
decreases by~1 the support of the edges in $\{\{u,z\}, \{v,z\} \mid z \in
N_G(u) \cap N_G(v)\}$ (while updating their buckets). As previously observed, 
this can be done by looking at the smallest-degree endpoint of $e$, in 
$O(\min\{\degr(u),\degr(v)\})$ time. It returns as $t_G$ the maximum among the
supports of the extracted edges. Since each edge is removed once, the total cost
is $O(m + \sum_{\{u,v\} \in E(G)} \min\{\degr(u),\degr(v)\}) = O(m \, \alpha_G)$
time and $O(m)$ space.
\end{proof}

We remark that a seemingly equivalent result is obtained with different techniques in~\cite{burkhardt2018bounds} using the ``average degeneracy'' of the graph in, where the average degeneracy is  $\frac{1}{|E(G)|}\sum_{\{u,v\} \in E(G)} \min\{\degr(u),\degr(v)\}$.

\begin{lemma}
\label{lem:order-decomposition}
\label{prop:p1}
There is an algorithm to compute the truss order of an undirected graph $G$ with $n$ nodes and $m$ edges in $O(f(m,n))$ time iff there is an algorithm to compute the truss decomposition in $O(f(m,n))$ time, for a polynomial $f(m,n) = \Omega(m)$.
\end{lemma}
\begin{proof}
As the other implication is easy, let us discuss how to obtain the truss decomposition from the truss order. 
First, let us consider the gadget $G_x$, that is a graph with $2x$ nodes:  
$G_x$ is build by adding $x$ nodes to a clique $K_x$, where the $i$th added node ($1 \leq i \leq x$) is connected arbitrarily to $i$ nodes in $K_x$. It can be easily noted how for each value $j \in \{0,\ldots, x\}$ there is at least one edge in $G_x$ with trussness $j$.
Now, compute the degeneracy $d_G$ of $G$, which can be done in $O(m)$ time, and recall that $t_G\le d_G = O(\sqrt{m})$.
Then, let $G'$ be the disjoint union $\gmirr{2} \cup G_{d_G}$: as the number of nodes and edges $G'$ is $O(n)$ and $O(m)$ respectively, we can compute its truss order in $O(f(m,n))$ time.
Recalling that the trussness of each edge of $\gmirr{2}$ is even, and $G_{d_G}$ generates edges with all values of trussness up to $d_G$ (which are known beforehand), each edge in $G^{\times 2}$ will appear in the truss order of $G'$ between two consecutive edges of $G_{d_G}$, of which at least one has even trussness, which will be the exact trussness of the edge. After that, 
we divide the value by~2 to obtain the exact trussness of each edge in $E(G)$.
\end{proof}

Lemma~\ref{lem:order-decomposition} implies that we cannot find a more efficient algorithm for the trussness $t_G$ using the truss order, as the latter costs as much as computing $t_G$ (see Lemmas~\ref{lem:tritrussmd} and~\ref{lem:conditional-tritrussmd}).

Since $m \, \alpha_G = \Theta(n^3)$ in the worst case, Theorem~\ref{thm:subcubic} implies that improving the worst-case cost in Lemma~\ref{lem:tritrussmd} to significantly less than $O(m \, \alpha_G)$ time using combinatorial algorithms is quite hard. However, since trussness $t_G$ is also a parameter for complexity analysis, one could hope to replace $\alpha_G$ with $(t_G+1)$ in Lemma~\ref{lem:tritrussmd}, thus getting $O(m (t_G+1))$ time. Not even this is possible, as $t_G = 0$ would mean that we can test if $G$ is triangle-free. In summary, we immediately derive the following conditional lower bound.

\begin{lemma}
\label{lem:conditional-tritrussmd}
Given any undirected graph $G$ with $m$ edges, arboricity $\alpha_G$ and trussness $t_G$, 
triangle counting/listing and graph trussness
cannot be computed by combinatorial algorithms in either $\tilde{o}(m \, \alpha_G)$ time or $O(m (t_G+1))$ time in the worst case, unless BMM is truly subcubic.
\end{lemma}

By Lemma~\ref{lem:tritrussmd} and Lemma~\ref{lem:conditional-tritrussmd}, it makes sense to investigate the problem of approximating the graph trussness $t_G$.
\end{document}